\documentclass[11pt]{article} 
\usepackage{template}
\usepackage[utf8]{inputenc}
\usepackage[english]{babel}

\usepackage{amsmath,amssymb,natbib,graphicx,url,algorithm2e}
\usepackage{times}
\usepackage{amsthm}
\usepackage{nicefrac}
\newtheorem{theorem}{Theorem}
\newtheorem{lemma}[theorem]{Lemma}
\newtheorem{corollary}[theorem]{Corollary}
\newtheorem{remark}[theorem]{Remark}
\newtheorem{definition}[theorem]{Definition}
\usepackage{color}
\usepackage[dvipsnames]{xcolor}
\definecolor{mydarkblue}{rgb}{0,0.08,0.45}
\usepackage[colorlinks=true,
    linkcolor=mydarkblue,
    citecolor=mydarkblue,
    filecolor=mydarkblue,
    urlcolor=mydarkblue]{hyperref}  
\usepackage{caption}
\usepackage{subcaption}
\usepackage{natbib}

\usepackage[inline, shortlabels]{enumitem}  

\begin{document}

\title{Finite Regret and Cycles with Fixed Step-Size via Alternating Gradient Descent-Ascent}
\author{
  James P. Bailey \\
  Singapore University\\
  of Technology and Design\\
  \small \texttt{james\_bailey@sutd.edu.sg}
  \And
  Gauthier Gidel \\
  Mila, Element AI, \\
  University of Montreal\\
  \small \texttt{gidelgau@mila.quebec}
  \And
  Georgios Piliouras \\
  Singapore University\\
  of Technology and Design\\
  \small \texttt{georgios@sutd.edu.sg}
}

\maketitle

\begin{abstract}
Gradient descent is arguably one of the most popular online optimization methods with a wide array of applications.  However, the standard implementation where agents simultaneously update their strategies yields several undesirable properties; strategies diverge  away from equilibrium and regret grows over time. In this paper, we eliminate these negative properties by introducing a different implementation to obtain finite regret via arbitrary fixed step-size. We obtain this surprising property by having agents take turns when updating their strategies. In this setting, we show that an agent that uses gradient descent obtains bounded regret -- regardless of how their opponent updates their strategies.  Furthermore, we show that in adversarial settings that agents' strategies are bounded and cycle when both are using the alternating gradient descent algorithm.
\end{abstract}
\vspace{-.5cm}
\begin{figure}[!htb]
    \centering
    \begin{subfigure}{.5\textwidth}
          \centering\captionsetup{justification=centering}
          \includegraphics[scale=.45]{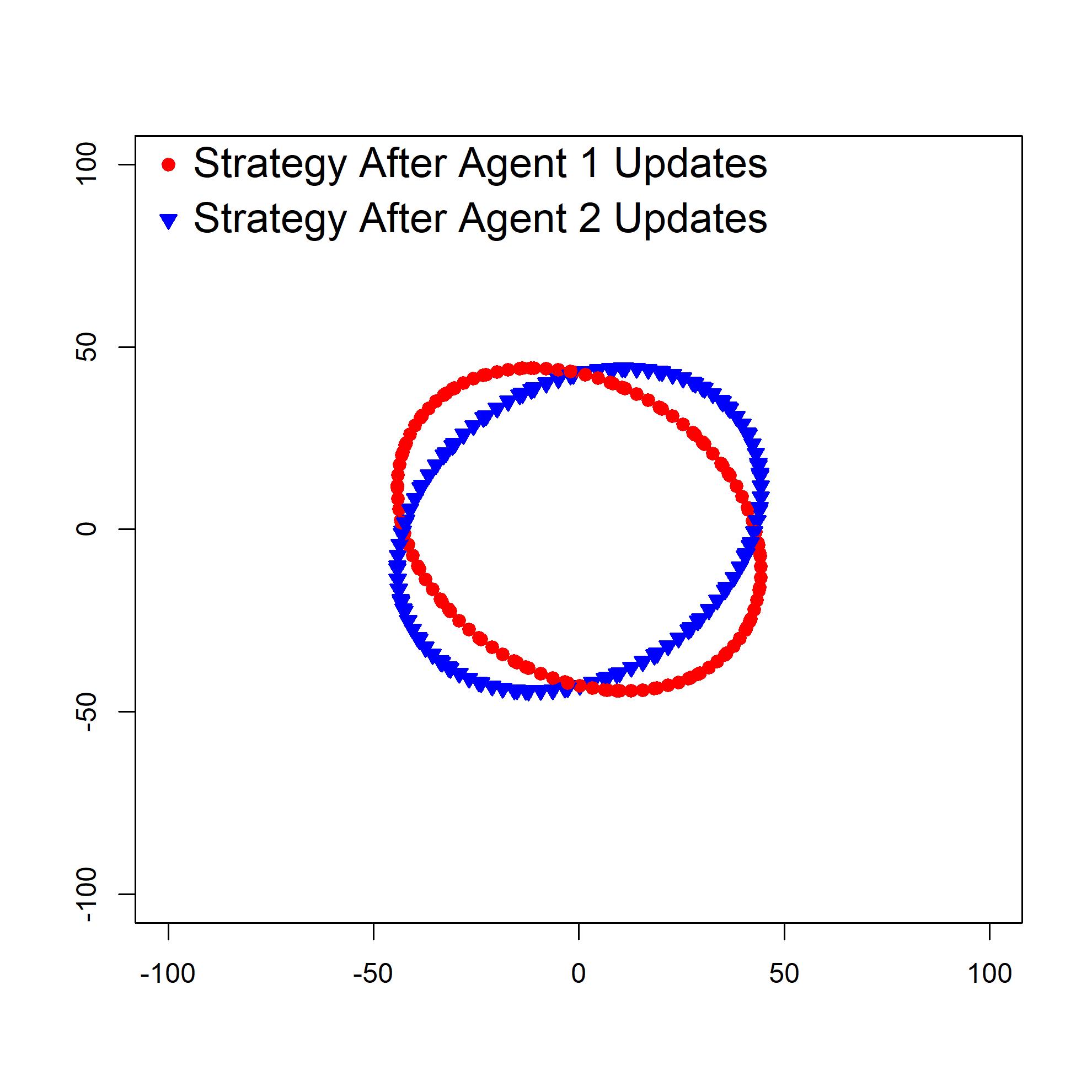}
          \caption{Player strategies cycle around the max-min equilibrium.}
          \label{fig:OpeningPoints}
    \end{subfigure}%
    \begin{subfigure}{.5\textwidth}
          \centering\captionsetup{justification=centering}
          \includegraphics[scale=.45]{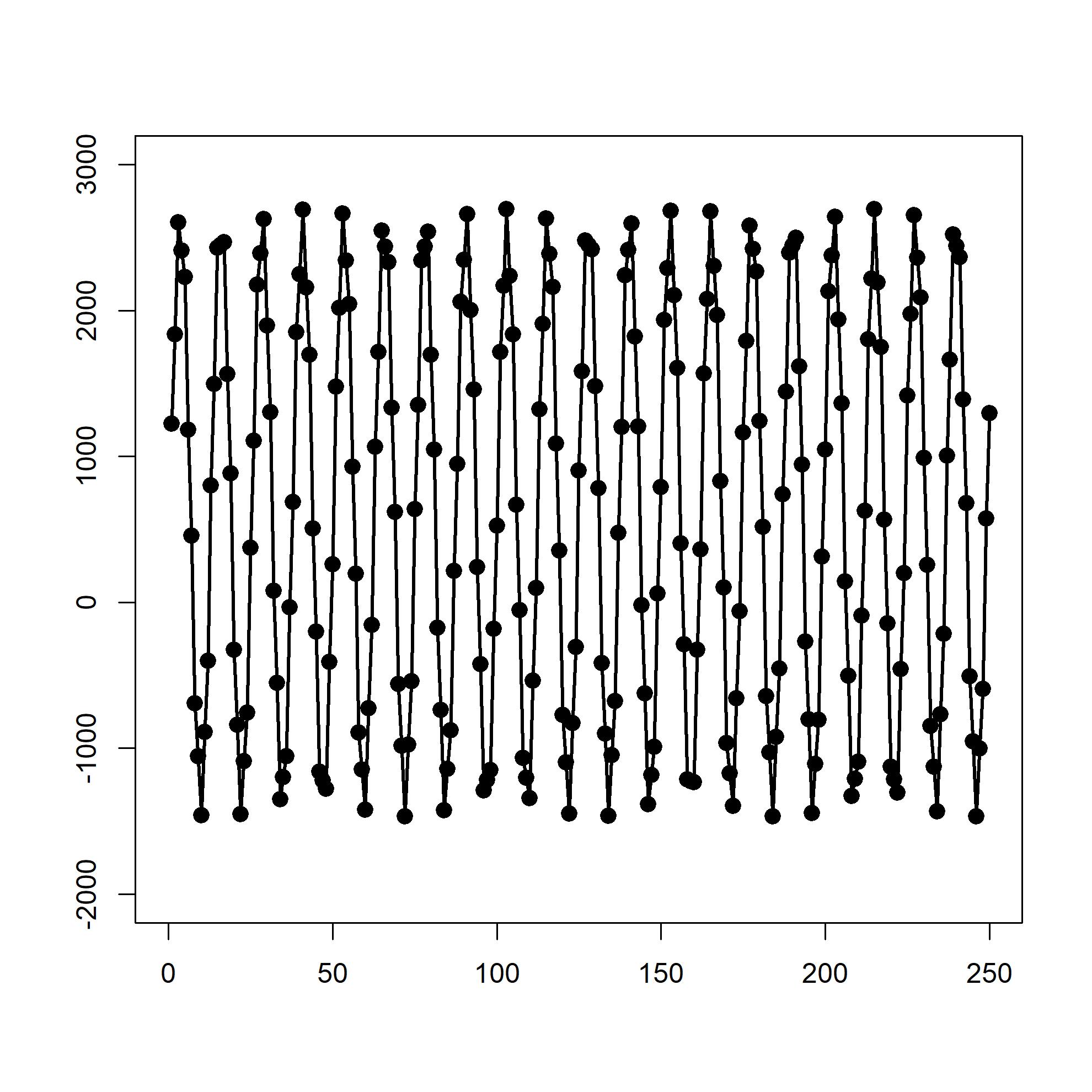}
          \caption{Agent 1's regret remains bounded while oscillating back and forth.}
          \label{fig:OpeningUtility}
    \end{subfigure}
    \caption{125 Iterations of Alternating Gradient Descent-Ascent applied in a zero-sum game with initial condition $(x_1^0,x_2^0)=(35,35)$, $A=[1]$ and learning rate $\eta_1=\eta_2=1/2$.}
    \label{fig:Opening}
\end{figure}
\section{Introduction}


Zero-sum games and more generally max-min optimization are amongst the most well studied settings in game theory. Dating back to classic work of  \citet{Neumann1928}, that initiated the field of game theory as a whole, it is well understood that zero-sum games admit a ``solution". The safety level that each agent can guarantee for themselves, if their were forced to commit to their strategies first, is exactly equal to the best case payoff they will get if they play second, with full information, against a rational opponent.  This fact that the ordering of the agents does not matter is captured by arguably the most famous aphorism in game theory, max-min is equal to min-max.

Despite the classically resolved issue of equilibrium computation in zero-sum games, the question of analyzing  dynamics  in zero-sum games is much less understood. Possibly, the most well known result in the area is that regret minimizing dynamics converge in a time average sense to max-min equilibria (e.g., \citet{Freund99schapire:adaptive}). However, up until recently, the day-to-day behavior of standard classes of online learning dynamics were not understood. For example, does the day-to-day behavior converge to equilibrium, does it diverge away from it, or does it cycle at a fixed distance from it? The answer to the above questions turned out to be, Yes, Yes and Yes! 
 Or, to be more precise, the answer depends critically on the choice of the dynamics.
 
 When studying dynamics in continuous-time, e.g., the continuous time-analogue of Multiplicative Weights Update, replicator dynamics, the dynamics "cycle" around at a constant Kullback-Leibler divergence from equilibrium as shown by \citet{Soda14}. 
 In fact, this result generalizes for all continuous-time variants of all Follow-the-Regularized Leader (FTRL) algorithms \citep{GeorgiosSODA18}. 
 Moreover, these dynamics have bounded (total/aggregate) regret \textit{in arbitrary games}. 
 This is an impressive level of regularity and predictability of the dynamics, which despite not being equilibrating, allow us to make strong predictions about their day-to-day behavior.  Behind this clockwork kind of regularity lies the fact that these dynamics are  Hamiltonian \citep{BaileyAAMAS19}.
 As in the case, e.g., of planetary orbits, or a pendulum, there is a lot of hidden structure in the motion, laws that bind and control the evolution of all particles. 
 
 Unfortunately, this level of regularity comes at a cost of using a continuous-time model. This is of course a simplifying modelling assumption. It does not capture the reality of how games, economic competition is played out in practice. More importantly, it fails to capture the reality of some modern engineering applications, such as Generative Adversarial Networks (GANs) \citep{goodfellow2014generative}, where online training algorithms compete against each other to improve two (opposing) AI algorithms. Hopefully, we could just naively discretize the aforementioned dynamics and their behavior would for the most part stay intact. Unfortunately, this is far from the truth.

 \citet{BaileyEC18} first proved that for all Follow-the-Regularized-Leader algorithms (e.g. Gradient Descent or Multiplicative Weights) diverge away from the Nash equilibrium in zero-sum games. This is a robust finding that holds regardless of the step-size that the agents are using, even if the agents are using different or shrinking step-sizes, or even if they are using different dynamics (i.e., mix-and-match regularizers). Finally, it even extends to network generalizations of zero-sum games.
 The proof by picture is as follows: If gradient descent-ascent in continuous-time moves along a Euclidean ball centered at the equilibrium then the naive discretization takes a discrete, non-negligible step along the tangent. Hence, the distance (radius) from the equilibrium grows and we keep moving away from the equilibrium.
  Even more distressingly, not only are equilibria unstable but furthermore the dynamics are formally chaotic as small perturbations of initial conditions are amplified exponentially fast \citep{cheung2019vortices}. In a nutshell, by discretizing gradient descent, moving from a differential equation to an actual implementable algorithm, all system regularity is lost. The discrete and continuous time behavior may only differ by a little bit in each step, but these errors snowball fast.
 
 Given the unstable nature of the naive, standard discretization of online gradient descent, several different training algorithms have been suggested which have been shown to converge provably to Nash equilibria in zero-sum games such as the extragradient method \citep{korpelevich1976extragradient} and its variants \citep{gidel2019a,mertikopoulos2019optimistic}, optimistic mirror descent \citep{rakhlin2013optimization,daskalakis2018training,daskalakis2018last} and some other methods using negative momentum or second order information \citep{gidel2019negative,balduzzi2018mechanics,2019arXiv190602027A}. Going back to the picture of simultaneous gradient descent-ascent as a tangent to a ball centered at equilibrium, these approaches "corrupt" the dynamics, so that the discrete-step is now facing the interior of the ball, more like a chord than a tangent, decreasing the distance from equilibrium and forcing convergence in the long run.
 
 We take a different approach when it comes to discretizing the system dynamics. We ask, as von Neumann did for equilibrium computation, does the ordering of the agents matter? What if the min and max agents did not update their behavior simultaneously but instead they took turns. This is actually common practice in training neural networks as no extra memory is needed to hold the previous state/parameters of any network. Even for economic competition in markets, this is a rather reasonable model with firms taking turns responding to the last move of the competition. Could it be that this standard alternating gradient descent-ascent implementation recovers some of the impressive regularities of the continuous-time model and if so to what extent?

\subsection*{Our Contributions} 
We study the behavior of gradient descent with fixed step-size in unconstrained two-agent zero-sum games. 
In a twist on the standard theory of online learning, we consider agents that take turns updating their strategies. 
We establish that if an agent uses gradient descent with arbitrary fixed step-size when agents are sequentially updating their strategies, then she obtains bounded regret (Theorem \ref{thm:bounded_regret}) as depicted in Figure \ref{fig:OpeningUtility}. 
Moreover, Theorem \ref{thm:bounded_regret} holds regardless of how her opponent updates his strategies and therefore the result immediately extends to non-zero-sum games. 
We establish this surprising property by showing that an agent's distance from optimality changes proportionally to her payoff in any given iteration (Lemma \ref{lem:change}).  
This allows us to compute both the regret and utility of an agent with only knowledge of the first and last strategy she used, regardless of how her opponent updates his strategies. 
The bound on regret quickly follows by considering the worst-case final strategy. 

We further explore the asymptotic properties of alternating gradient descent specifically in the setting of zero-sum games.  
We show that when agents use gradient descent sequentially that the strategies approximately cycle (Theorem \ref{thm:poincar_rec}) as depicted in Figure \ref{fig:OpeningPoints}.  
More formally, alternating gradient descent admits Poincar\'{e} recurrence in the setting of two-agent zero-sum games. 
Theorem~\ref{thm:poincar_rec} is established in two parts:
First, we show that the alternating gradient descent algorithm approximately preserves the distance to the equilibrium (Theorems~\ref{thm:bounded_orbits} and~\ref{thm:lower_bound}). 
Second, we show that the algorithm preserves volume when updating a measurable set of strategies. 
Together, these two properties imply recurrence. 

\section{Preliminaries}
A two-agent zero-sum game consists of two agents ${\cal N}=\{1,2\}$ where agent $i$ selects a strategy from $\mathbb{R}^{k_i}$.  
Utilities of both agents are determined via a payoff matrix $A\in \mathbb{R}^{k_1 \times k_2}$. 
Given that agent $i$ selects strategy $x_i\in \mathbb{R}^{k_i}$, agent 1 receives utility $\langle x_1, Ax_2\rangle$ and agent 2 receives utility $-\langle x_1, Ax_2\rangle$. 
Naturally, both agents want to maximize their payout resulting in the following max-min problem:
\begin{align}\tag{Zero-Sum Game}\label{eqn:ZeroSum}
    \max_{x_1\in \mathbb{R}^{k_1}}  \min_{x_2\in \mathbb{R}^{k_2}} \{ x_1\cdot A x_2 \}
\end{align}

\subsection{Gradient Descent with Simultaneous Play}
In many applications of game theory, agents know neither the payoff matrix nor their opponent's strategy.  
Instead, agents repeatedly play the zero-sum game while updating their strategies iteratively. 
One of the most common methods for updating strategies is the gradient descent algorithm.  
In gradient descent, an agent looks at her payout in the previous iteration and then updates her previous strategies by moving in a most beneficial direction.  
In the setting of (\ref{eqn:ZeroSum}), this corresponds to 
\begin{equation}\tag{SimGD}\label{eqn:GDSimultaneous}
    \begin{aligned}
      x_1^{t+1} &= x_1^t + \eta_1 Ax_2^t\\
      x_2^{t+1} &= x_2^t - \eta_2 A^\intercal x_1^t.
    \end{aligned}
\end{equation}
where $\eta_i$ corresponds to agent $i$'s step-size. 
The larger the step-size, the more rapidly an agent responds to information from previous iterations.
Gradient descent is often implemented with time variant step-sizes -- most commonly with $\eta_i^t\in \Theta(1/\sqrt{t})$. 
However, in this paper we focus on fixed step-sizes. 

In this formulation, agents simultaneously update their strategy.  That is, $x_1^t$ and $x_2^t$ are played at the same time.  
As a result, the cumulative utility of (\ref{eqn:GDSimultaneous}) (or any simultaneous update algorithm) for agent 1 after $T$ iterations is 
\begin{align}
    \sum_{t=0}^T \langle x_1^t,Ax_2^t\rangle.
\end{align}

\subsection{Gradient Descent with Alternating Play}
In many application of game theory, agents do not update their strategies until they see a change in the system.  
In the case of two-agent games, this corresponds to agents updating their strategies sequentially, i.e., agent 1 updates her strategy, then agent 2 updates his strategy, then agent 1 updates her strategy and so on. 
In the setting of gradient descent, this corresponds to 
\begin{equation}\tag{AltGD}\label{eqn:GDAlternating}
    \begin{aligned}
      x_1^{t+1} &= x_1^t + \eta_1 Ax_2^t\\
      x_2^{t+1} &= x_2^t - \eta_2 A^\intercal x_1^{t+1}.
    \end{aligned}
\end{equation}
Computing the total utility when agents alternate their updates is slightly different. 
Agent 1 plays strategy $x_1^t$ against $x_2^t$ when agent 2 updates his strategy and plays $x_1^{t+1}$ against $x_2^t$ when agent 1 updates her strategy.  
This results in the following cumulative utility after agent $1$ updates her strategy $T$ times.
\begin{align}
    \sum_{t=0}^T \langle x_1^{t+1}+x_1^t,Ax_2^t\rangle
\end{align}

\subsection{Regret}

The standard way of measuring the performance of an algorithm is by a notion known as regret.  
Regret compares the total utility gained by a fixed strategy $x_1$ to the utility agent 1 receives by using an algorithm. 
In the case of simultaneous updates, as in (\ref{eqn:GDSimultaneous}), regret is formally given by
\begin{align}\tag{Regret for Alternating Play}
    \langle x_1,\sum_{t=0}^T Ax_2^t\rangle - \sum_{t=0}^T \langle x_1^t,Ax_2^t\rangle 
\end{align}
where the second term corresponds to the utility agent 1 received by using (\ref{eqn:GDSimultaneous}) and first term corresponds to the utility she would of received if she played $x_1$ on every iteration (assuming agent $2$ still uses the strategies $\{x_2\}_{t=1}^T$).
In the case of constrained optimization, regret is typically evaluated where $x_1$ is the best fixed strategy, i.e., the optimizer of $\langle x_1, \sum_{t=1}^TAx_2^t\rangle$. 
However, in unconstrained optimization there is rarely an optimizer to this expression. 

When agents update sequentially, regret is computed slightly differently.  
As discussed in the previous section, agent 1 sees the strategy $x_2^t$ twice -- once when agent 1 updates and once when agent 2 updates.  
As a result, agent 1's regret when updating sequentially is 
\begin{align}\tag{Regret for Sequential Updates}
    \langle 2x_1,\sum_{t=0}^T Ax_2^t\rangle - \sum_{t=0}^T \langle x_1^{t+1}+x_1^t,Ax_2^t\rangle.
\end{align}

Typically an algorithm is said to perform well if its regret is bounded above by a sublinear function with respect to any fixed strategy. 
If regret grows at a rate of $o(T)$ with respect to a fixed strategy, then the average regret grows as at a rate of $o(1)$ and, in the limit, the algorithm performs no worse on average as that fixed strategy.

\subsection{Continuous-Time Gradient Descent: A Motivation for Alternating Play}

The primary motivation for this paper is the continuous-time analogue of gradient descent.  
In particular, the integration technique used to obtain (\ref{eqn:GDAlternating}) from the continuous-time analogue well approximates continuous-time and therefore offer similar guarantees for behavior in the system. 
The continuous-time analogue of (\ref{eqn:GDSimultaneous}) and (\ref{eqn:GDAlternating}) is
\begin{equation}\tag{ContinuousGD}\label{eqn:GDContinuous}
    \begin{aligned}
      x_1(t) &=  x_1(0) + \eta_1\int_0^tAx_2(s)ds\\
      x_2(t) &= x_2(0) - \eta_2\int_0^tA^\intercal x_1(s)ds \,,
    \end{aligned}
\end{equation}
where $\eta_i$ denotes the learning rate used by agent $i$. 
(\ref{eqn:GDSimultaneous}) is obtained from (\ref{eqn:GDContinuous}) via Euler integration, i.e., $x_i^t$ is simply the first order approximation of $x_i(t)$ from the point $x_i(t-1)$. 

\cite{GeorgiosSODA18} showed that (\ref{eqn:GDContinuous}) cycles around the equilibrium of the system on convex orbits.  
Therefore, (\ref{eqn:GDSimultaneous}) should diverge from the equilibrium since it is the first order approximation of  (\ref{eqn:GDContinuous}). 
Indeed, this is first formally shown for a more general class of update rules including gradient descent and multiplicative weights in a paper by \citet{BaileyEC18} and for gradient descent in unconstrained bilinear games by~\citet{gidel2019negative}.
Moreover, \cite{GeorgiosSODA18} showed that \eqref{eqn:GDContinuous} has a surprising property; if agent 1 uses \eqref{eqn:GDContinuous} and agent 2 uses any continuous update rule, then agent 1 obtains bounded regret. 
Recall from the previous section that sublinear regret is the desired property in order to obtain optimal convergence results for online learning.  
Therefore, \eqref{eqn:GDContinuous} obtains impressive regret and convergence guarantees. 

Unfortunately, continuous-time algorithms are difficult to run and online optimization typically relies on discrete-time algorithms. 
Regrettably, standard discrete-time algorithms fall short relative to their continuous-time analogues;
it has long been believed that (\ref{eqn:GDSimultaneous}) with fixed step-size has linear regret and therefore offers no nice long-term guarantees. 
It can however obtain $\Theta(\sqrt{T})$ regret by employing decaying step-sizes or after carefully selecting a fixed step-size with prior knowledge of how long the algorithm will be used~\citep{foster2016learning,hazan2016introduction}. 
More recently, \cite{BaileyArxiv19} showed that (\ref{eqn:GDSimultaneous}) with arbitrary  fixed step-size also obtains $\Theta(\sqrt{T})$ regret in bounded 2-dimensional zero-sum games and offered experimental evidence to suggest the result carries over to higher dimensions.
However, even these improved guarantees fall short of the bounded regret obtained by (\ref{eqn:GDContinuous}). 

\cite{BaileyAAMAS19} recently offered an explanation for this contrast. 
They showed that (\ref{eqn:GDContinuous}) forms something known as a Hamiltonian system -- a common dynamical system studied in mathematical physics that explains things like planetary orbits and oscillating springs. 
These systems conserve energy -- in the case of (\ref{eqn:GDContinuous}), energy corresponds to the combined norm of the strategies $\nicefrac{||x_1||^2}{\eta_1}+\nicefrac{||x_2||^2}{\eta_2}$ partially explaining the cyclic behavior found by \cite{GeorgiosSODA18}. 
However, Euler integration is well-known to be a poor estimator of Hamiltonian systems and it is therefore unsurprising that (\ref{eqn:GDContinuous}) differs greatly from (\ref{eqn:GDSimultaneous}). 

To obtain behavior that is similar to (\ref{eqn:GDContinuous}), we must use an integration technique that better preserves the dynamics of the original system. 
Fortunately, there is a particular class of integrators known as symplectic integrators that are well-known for their ability to approximate Hamiltonian systems~\citep{hairer2006geometric}. 
In particular, \cite{Hairer05} showed that symplectic integrators approximately preserve the energy of a Hamiltonian system for exponentially long periods of time relative to the inverse of the fixed step-size.  
In the setting of Gradient Descent, that means that we could approximately preserve the energy in (\ref{eqn:GDContinuous}) for arbitrarily long periods of time by applying a symplectic integrator with sufficiently small step-size. 

Regrettably, many symplectic integrators would require agents to coordinate when updating their strategies.  
By the very nature of a zero-sum game, this coordination would be unnatural and could only be applied in artificial settings such as GANS. Moreover, in the context of machine learning, we would like to use a sufficiently large step-size in order to reduce the training time. A vanishing step size with an increasing horizon would lead to a prohibitively slow training method.

It is straightforward to show that (\ref{eqn:GDAlternating}) is obtained from (\ref{eqn:GDContinuous}) via the symplectic integration technique known as Verlet integration or leapfrogging. 
By the work of \cite{Hairer05}, we therefore expect that (\ref{eqn:GDAlternating}) with fixed step-size should behave similarly to (\ref{eqn:GDContinuous}) -- at least for exponentially long periods of time relative to the step-size.
Indeed we actually show a stronger result; 
(\ref{eqn:GDAlternating}) with a fixed step-size has the same guarantees of (\ref{eqn:GDContinuous}) forever. 
Specifically, if agent 1 uses  (\ref{eqn:GDAlternating}) with arbitrary fixed step-size then she obtains bounded regret regardless of how her agent's opponent updates. 
Moreover, if both agents use (\ref{eqn:GDAlternating}) then the quantity $\big(\nicefrac{||x_1^t||^2}{\eta_1}+\nicefrac{||x_2^t||^2}{\eta_2} + \langle x_1^t, Ax_2^t\rangle \big)$ is preserved and the strategies $\{x_1^t,x_2^t\}_{t=0}^\infty$ cycle for $\sqrt{\eta_1\eta_2} \leq \tfrac{2}{\|A\|}$, allowing step-sizes that do not vanish with an infinite horizon. 

We proceed by proving bounded regret in Section \ref{sec:regret} and recurrent behavior in Section \ref{sec:cycle}.

\section{Bounded Regret with Fixed Step-Size in Gradient Descent.}\label{sec:regret}

In this section, our focus we will be on the regret generated by an agent playing according to \eqref{eqn:GDAlternating}. Interesting, our result holds no matter how the opponent updates its strategy. This general setting is particularly interesting because it is able to model an environment with only partial information where the agents might even not know that they are playing a game.

Before stating the main theorem of this section, we present a lemma that provides an interpretation of each agent's payoff in terms of energy fluctuation. The norm of an agent's strategy, rescaled by its step-size, can be seen as an energy that varies proportionally to its payoff.

\begin{lemma}\label{lem:change}When agent $1$ updates via (\ref{eqn:GDAlternating}), the size of $x_1^t$ increases proportionally to agent $1$'s payoff since her update in iteration $t$.  Formally,
\begin{align}
    \frac{||x_1^{t+1}||^2-||x_1^t||^2}{\eta_1}&=\langle x_1^{t+1}+x_1^t, Ax_2^t\rangle.\label{eqn:agent1}
\end{align}
Similarly, when agent $2$ updates via  (\ref{eqn:GDAlternating}), the size of $x_2^t$ increases proportionally to agent $2$'s payoff since his update in iteration $t$.  Formally,
\begin{align}
    \frac{||x_2^{t+1}||^2-||x_2^t||^2}{\eta_2}&=-\langle x_1^{t+1}, A(x_2^{t+1}+x_2^t)\rangle\label{eqn:agent2}.
\end{align}
\end{lemma}

\begin{proof}
Recall the update rule for agent 1 that updates via \eqref{eqn:GDAlternating} is
\begin{align} 
    x_1^{t+1} &= x_1^t + \eta_1 Ax_2^t
\end{align}
Thus,
\begin{align}
    ||x_1^{t+1}||^2-||x_1^t||^2&=||x_1^t+\eta_1Ax_2^t||^2-||x_1^t||^2\\
    &= ||\eta_1Ax_2^t||^2+2\eta_1\langle x_1^t, Ax_2^t\rangle \\
    &=\eta_1\langle \eta_1Ax_2^t, Ax_2^t\rangle +2\eta_1\langle x_1^t, Ax_2^t\rangle \\
    &=\eta_1\langle x_1^{t+1}-x_1^t, Ax_2^t\rangle +2\eta_1\langle x_1^t, Ax_2^t\rangle \\
    &=\eta_1\langle x_1^{t+1}+x_1^t, Ax_2^t\rangle
\end{align}
completing the proof for agent 1.   
Similarly for agent 2, the updates via \eqref{eqn:GDAlternating} is 
\begin{align} 
    x_2^{t+1} &= x_2^t - \eta_2 A^\intercal x_1^{t+1}
\end{align}
and thus, we get,
\begin{align}
    ||x_2^{t+1}||^2-||x_2^t||^2&=||x_2^t-\eta_2A^\intercal x_1^{t+1}||^2-||x_2^t||^2\\
    &= ||\eta_2A^\intercal x_1^{t+1}||^2-2\eta_2\langle  x_1^{t+1},Ax_2^t\rangle \\
    &=\eta_2\langle A^\intercal x_1^{t+1}, \eta_2A^\intercal x_1^{t+1}\rangle -2\eta_2\langle  x_1^{t+1},Ax_2^t\rangle \\
    &=\eta_2\langle A^\intercal x_1^{t+1}, x_2^{t}-x_2^{t+1}\rangle-2\eta_2\langle  x_1^{t+1},Ax_2^t\rangle\\
    &=-\eta_2\langle x_1^{t+1}, A(x_2^{t+1}+x_2^t)\rangle
\end{align}
completing the proof of the lemma.
\end{proof}

From this lemma, an explicit bound on the regret of an agent that uses \eqref{eqn:GDAlternating}
 easily follows 
independently of what update rule her opponent uses. 
This result, is a significant improvement in comparison to the $\Theta(\sqrt{T})$ regret of~\eqref{eqn:GDSimultaneous} and surprisingly matches the result on continuous-time gradient descent provided by~\citet[Thm.~3.1]{GeorgiosSODA18}.

\begin{theorem}\label{thm:bounded_regret} If agent 1 uses (\ref{eqn:GDAlternating}) with an arbitrary fixed step-size $\eta_1$, then she obtains bounded regret with respect to any fixed strategy $x_1$ regardless of how her opponent updates his strategies. 
\end{theorem}

\begin{proof}
    Agent 1's regret with respect to strategy $x_1$ is
    \begin{align}
    \langle 2x_1, \sum_{t=0}^TAx_2^t\rangle -\sum_{t=0}^T \langle x_1^{t+1}+x_1^t,Ax_2^t\rangle &=
    \frac{\langle 2x_1, x_1^{T+1}-x_1^0\rangle}{\eta_1} -\sum_{t=0}^T \frac{||x_1^{t+1}||^2-||x_1^t||^2}{\eta_1}\\
    &=  \frac{\langle 2x_1, x_1^{T+1}-x_1^0\rangle}{\eta_1}-\frac{ ||x_1^{T+1}||^2- ||x_1^{0}||^2}{\eta_1}\\
    &=\frac{\langle 2x_1-x_1^{T+1},x_1^{T+1}\rangle- \langle 2x_1-x_1^0,x_1^0\rangle}{\eta_1}\label{eqn:Regret}\\
    &\leq \frac{\langle x_1^0- 2x_1,x_1^0\rangle + \|x_1\|^2}{\eta_1}
\end{align}
since the expression $x_1^{T+1} \mapsto \langle 2x_1-x_1^{T+1},x_1^{T+1}\rangle$ is maximized when $x_1^{T+1}=x_1$. 
\end{proof}
\begin{remark}
In Lemma \ref{lem:change} and Theorem \ref{thm:bounded_regret}, we have shown that we can compute agent $1's$ utility and regret with only knowledge of the first and last strategy she played.  
Moreover, neither proof requires knowledge of how the opponent updates or the payoff matrix.  
Therefore the results of Lemma \ref{lem:change} and Theorem \ref{thm:bounded_regret} extend to a variety of other settings including time-variant games (not even necessarily zero-sum) and one-agent systems with time-variant linear loss functions.
\end{remark}

\section{Recurrence and Bounded Orbits in Zero-Sum Games}\label{sec:cycle}
 \vspace{-.75cm}
\begin{figure}[!htb]
    \centering
    \begin{subfigure}{.5\textwidth}
          \centering\captionsetup{justification=centering}
          \includegraphics[scale=.45]{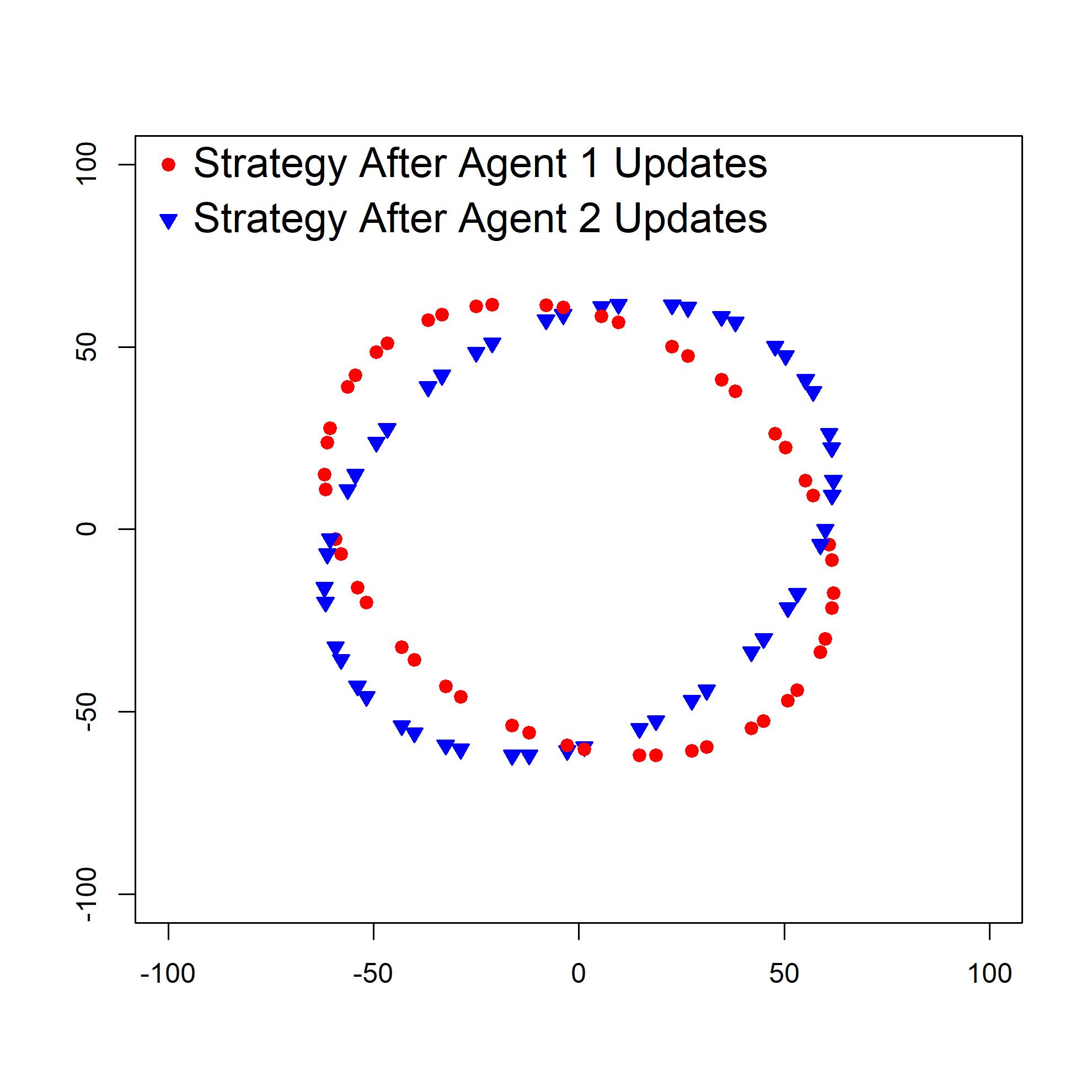}
          \caption{Strategies ``cycle" while approximating  $||x_1||^2+||x_2||^2=60$ with $\eta_1=\eta_2=1/2$.}
          \label{fig:Circular Orbit}
    \end{subfigure}%
    \begin{subfigure}{.5\textwidth}
          \centering\captionsetup{justification=centering}
          \includegraphics[scale=.45]{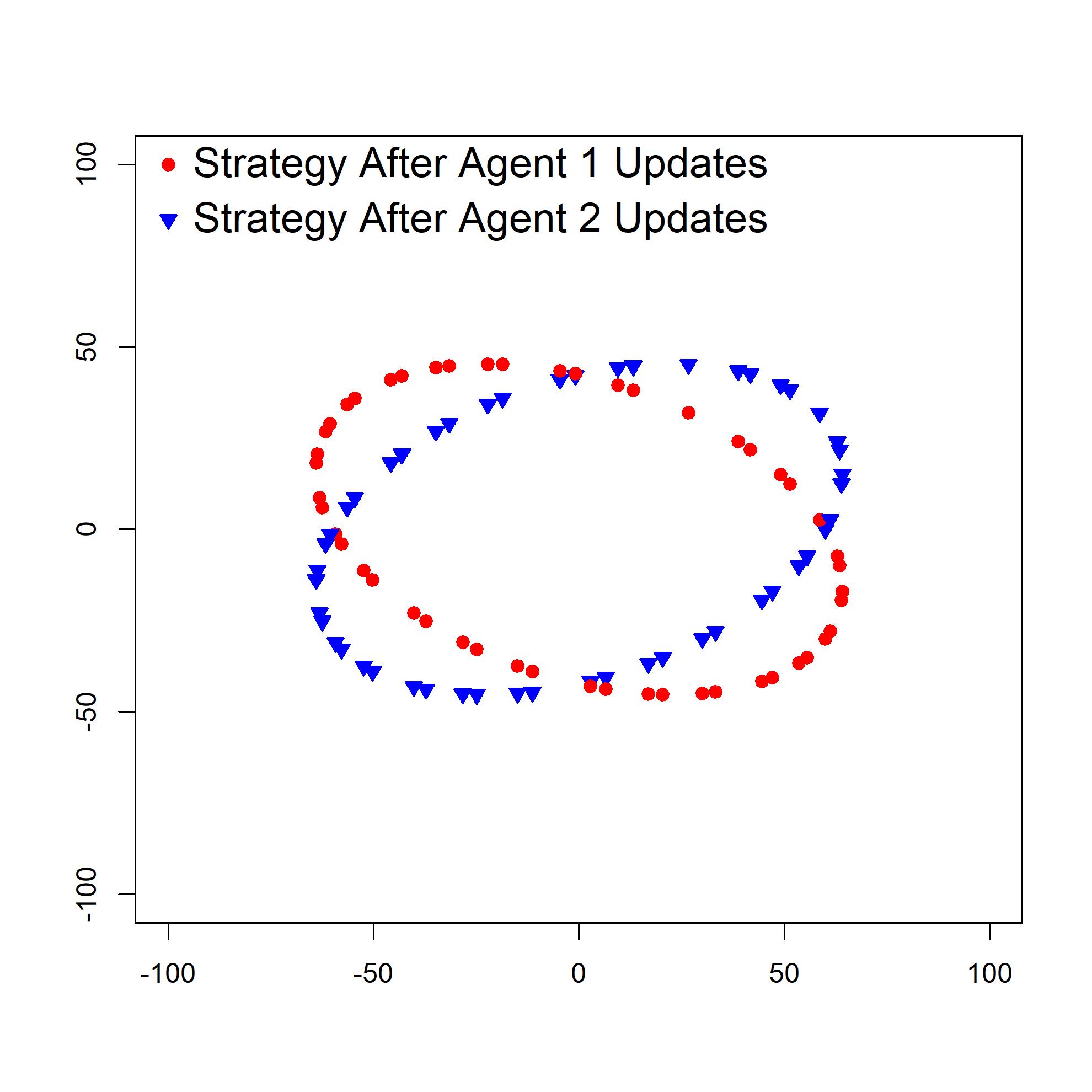}
          \caption{Strategies ``cycle" while approximating  $||x_1||^2+2||x_2||^2=60$  with $\eta_1=1, \eta_2=1/2$.}
          \label{fig:EllipticalOrbit}
    \end{subfigure}
    \caption{Initial strategy $(x_1^0,x_2^0)=(60,0)$ updated with 50 iterations of (\ref{eqn:GDAlternating}) with $A=[1]$. 
    }
    \label{fig:Ellipse}
\end{figure}

After having shown, in the previous section, that agents that play according to~\eqref{eqn:GDAlternating} have bounded regret we would like to investigate the asymptotic properties of their strategies. It has been recently proved that if each agent's strategy are updated though~\eqref{eqn:GDAlternating}, then the energy of the system $\nicefrac{\|x_1^t\|^2}{\eta_1} + \nicefrac{\|x_2^t\|^2}{\eta_2}$ is bounded above and below~\citep[Table 1]{gidel2019negative}. Thus, the strategies do not converge to the Nash equilibrium of the game. However, this boundedness, might indicate a cyclic behavior of the strategies. In the context of high dimensional dynamical system, this cyclic behavior is encompassed by the notion of Poincaré recurrence. Intuitively, a dynamical system is Poincaré recurrent if almost all trajectories return arbitrarily close to their initial position infinitely often. 


Indeed, as shown in Figure \ref{fig:Ellipse}, (\ref{eqn:GDAlternating}) appears to cycle. 
In this section, we formally prove the existence of Poincaré recurrence. 
Our analysis focuses on the strategies after both agents update -- i.e.,  $\{x_1^t,x_2^t\}_{t=0}^\infty$ -- as depicted by the blue triangles in Figure \ref{fig:Ellipse}.  
It also straightforward to extend our analysis to $\{x_1^{t+1}, x_2^t\}_{t=0}^\infty$ (depicted by the red circles in Figure \ref{fig:Ellipse}) through the same proof techniques.

More formally, in order to work with this notion of Poincaré recurrence we need to define a measure on $\mathbb{R}^d$. In the following, we will use the Lebesgue measure $\ell$. We can thus define the notion of a volume preserving transformation.

\begin{definition}[Volume Preserving Transformation~\citep{barreira}]\label{def:volume_preserving}
A volume preserving transformation is a measurable function $T:\mathbb{R}^d \to \mathbb{R}^d$ such that, for any open set $A \in \mathbb{R}^d$, we have $\ell(A) = \ell(T^{-1}(A))$.
\end{definition}
Note that $\ell(A)$ may be infinite. This notion of volume preserving transformation can be more generally defined on a orientable manifold. However, in this work, for simplicity, we will stick with the less general Definition~\ref{def:volume_preserving}. 
We can thus, state the Poincaré recurrence theorem.
\begin{theorem}[Poincaré Recurrence \citep{Poincare1890,barreira}]\label{thm:poincar_rec}  If a transformation preserves volume and has only bounded
orbits then it is Poincaré recurrent, i.e., for each open set there exist orbits that intersect this set infinitely often.
\end{theorem}

Furthermore, we can cover
 any region of $\mathbb{R}^{k_1+k_2}$ by countably many balls of radius $\epsilon/2$, and apply the previous theorem to each ball. We conclude that almost every point returns to within an $\epsilon$ of itself. Since $\epsilon>0$ is arbitrary, we conclude that almost every initial point is recurrent.
Formally, we will thus show the following corollary that states the (Poincaré) recurrence of~\eqref{eqn:GDAlternating}.

\begin{corollary}\label{cor:cycle}
For $\sqrt{\eta_1\eta_2}\leq \frac{2}{||A||}$ 
the (\ref{eqn:GDAlternating}) dynamic is Poincar\'{e} recurrent.  Moreover, for almost all initial conditions $(x_1^0,x_2^0)$ there exists an infinite sequence of time periods $\tau_n$ such that the $\lim_{n\to \infty} (x_1^{\tau_n},x_2^{\tau_n})=(x_1^0,x_2^0)$.
\end{corollary}

To show Corollary (\ref{cor:cycle}), it suffices to show that (\ref{eqn:GDAlternating}) has bounded orbits and that (\ref{eqn:GDAlternating}) preserves volume. 
In Section \ref{sec:bounded}, we show that the orbits are bounded if $\sqrt{\eta_1\eta_2}\leq \tfrac{2}{||A||}$.  
In Section \ref{sec:volume}, we show that volume is preserved regardless of the value of $\eta_1$ and $\eta_2$.

\subsection{Bounded Orbits}\label{sec:bounded}
\begin{figure}[!htb]
    \centering
    \begin{subfigure}{.5\textwidth}
          \centering\captionsetup{justification=centering}
          \includegraphics[scale=.45]{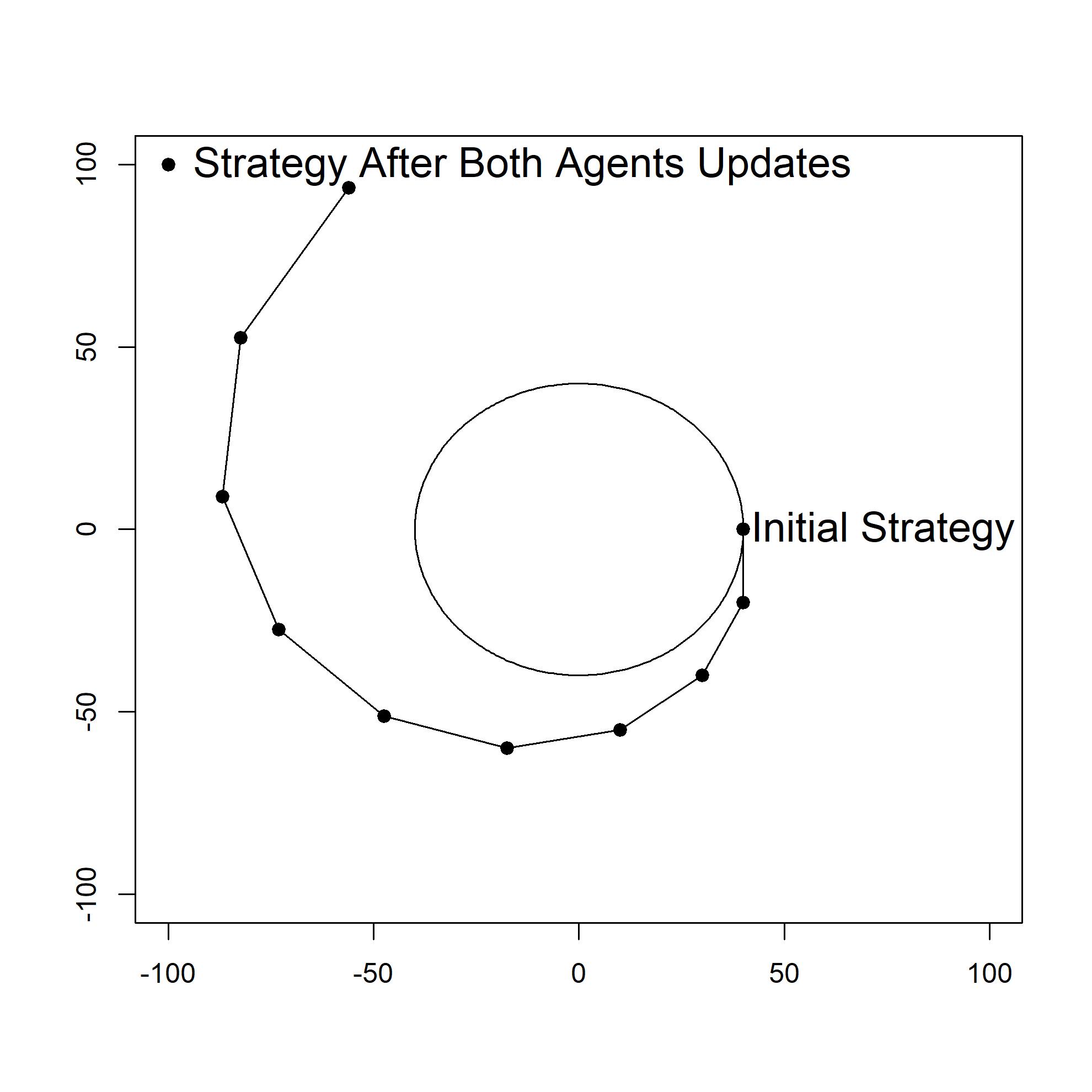}
          \caption{Strategies diverging in (\ref{eqn:GDSimultaneous}).\\\phantom{Hey}}
          \label{fig:EulerDiverges}
    \end{subfigure}%
    \begin{subfigure}{.5\textwidth}
          \centering\captionsetup{justification=centering}
          \includegraphics[scale=.45]{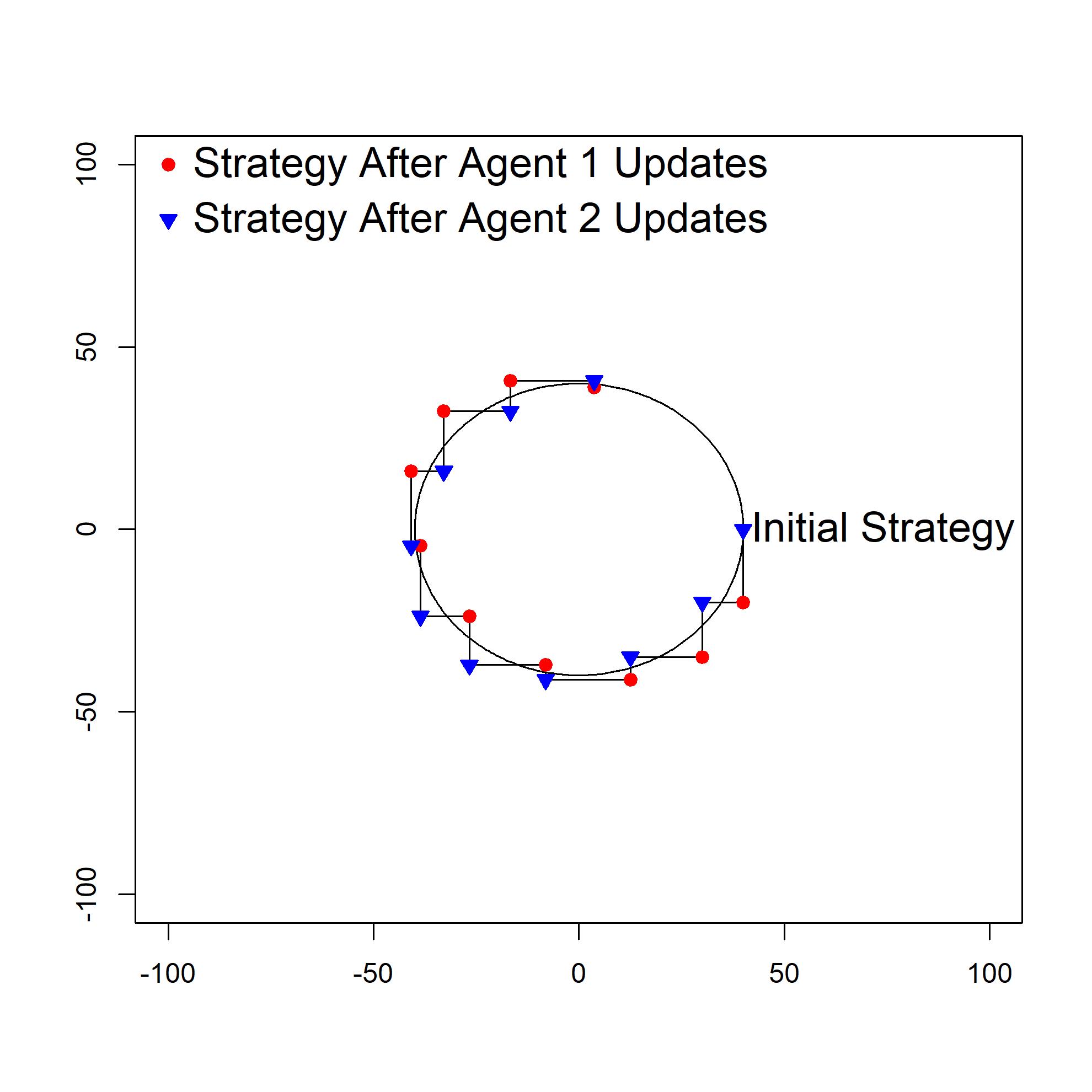}
          \caption{Strategies approximately preserving\\ energy  in (\ref{eqn:GDAlternating}).}
          \label{fig:VerletCycles}
    \end{subfigure}
    \caption{Initial strategy $(x_1^0,x_2^0)=(40,0)$ updated by 10 iterations of (\ref{eqn:GDSimultaneous}) and (\ref{eqn:GDAlternating}) with $A=[1]$, and $\eta_1=\eta_2=1/2$. 
    The circles denote  $\{x:||x_1||^2+||x_2||^2=40^2\}$.}
    \label{fig:Cycle}
\end{figure}
In this section, we prove that if both agents follow~\eqref{eqn:GDAlternating}, then their strategies are bounded. This result has been already proved by~\citet{gidel2019negative} using linear algebra arguments. However, in this section, we provide the following improvements: 
\begin{enumerate*}[series = tobecont, itemjoin = \quad, label=(\roman*)]
        \item We provide a new proof technique that is potentially generalizable to other geometries (\citet{gidel2019negative}'s proof using linear algebra argument heavily relies on the euclidean metric making it challenging to generalize to other geometries).
        \item For both the upper-bound and the lower-bound, this new proof technique has a clear interpretation in terms of energy conservation and provides an explicit dependence on the constants of the problem.
        \end{enumerate*}

The notion of conservation of energy we use in this section, is a perturbed version of the energy used in the continuous case (see Figure~\ref{fig:VerletCycles} for an illustration). If both agents use (\ref{eqn:GDAlternating}), the sum of their energies plus a payoff function is constant.
More precisely, we have the following theorem.
\begin{lemma}
    \label{lem:constant_quantity} If both agents use (\ref{eqn:GDAlternating}), we have that the following perturbed energy is constant,
\begin{equation}
 \frac{||x_1^{t}||^2}{\eta_1} + \frac{||x_2^{t}||^2}{\eta_2} + \langle x_1^{t}, Ax_2^{t}\rangle= \frac{||x_1^{0}||^2}{\eta_1} + \frac{||x_2^{0}||^2}{\eta_2} +\langle x_1^{0}, Ax_2^{0}\rangle \,.
\end{equation}
\end{lemma}
\begin{proof}
    Combining (\ref{eqn:agent1}) and (\ref{eqn:agent2}) of Lemma \ref{lem:change} yields
    \begin{align}
        \frac{||x_1^{t+1}||^2-||x_1^{t}||^2}{\eta_1}+\frac{||x_2^{t+1}||^2-||x_1^{t}||^2}{\eta_2}=\langle x_1^t,Ax_2^t\rangle - \langle x_1^{t+1},Ax_2^{t+1}\rangle.
    \end{align}
    Inductively, this implies the result claimed.
\end{proof}
If both agents' strategies are unidimentional, Lemma~\ref{lem:constant_quantity} has a geometric interpretation: the orbit of the joint strategy $\{(x_1^t,x_2^t)\}_{t=0}^\infty$ belongs to a conic section determined by the equation 
\begin{equation}
    \left(\frac{x_1}{\sqrt{\eta_1}}\right)^2 +\left(\frac{x_2}{\sqrt{\eta_2}}\right)^2 + a \cdot x_1 x_2 = 0  \,.
\end{equation}
We can show that this conic section is an ellipse if and only if $a^2 - \frac{4}{\eta_1\eta_2} \leq 0 $. Thus, for $\sqrt{\eta_1 \eta_ 2}\leq \frac{2}{a}$, the strategies are bounded. The eccentricity and the directions of the principal axis heavily depend on the values of $\eta_1$ and $\eta_2$. In Figure~\ref{fig:Ellipse}, we observe these elliptic trajectories in the joint strategies space for different values of the step-sizes.
This geometric argument can be generalized to strategies belonging to $\mathbb{R}^{k_1}\times \mathbb{R}^{k_2}$ using the singular vectors of $A$, however,
for simplicity, we provide a result in terms of weighted norms using a more concise proof.
\begin{theorem}[Bounded orbits] \label{thm:bounded_orbits} If both agents use (\ref{eqn:GDAlternating}) with arbitrary fixed step-sizes, we have that, for all $t\geq 0$,
\begin{equation}
     \left(1 - \frac{\sqrt{\eta_1 \eta_ 2} \|A\|}{2}\right) \left(\frac{||x_1^{t}||^2}{\eta_1} + \frac{||x_2^{t}||^2}{\eta_2}\right) \leq  \langle x_1^0,Ax_2^0\rangle + \tfrac{||x_1^{0}||^2}{\eta_1} + \tfrac{||x_2^{0}||^2}{\eta_2}.\label{eqn:upperbound}
\end{equation}
Thus, if they select their learning rates such that $\sqrt{\eta_1\eta_2} \leq \frac{2}{\|A\|}$, then their strategies are bounded.
\end{theorem}
\begin{proof}
Starting from Lemma~\ref{lem:constant_quantity}, we have that,
    \begin{align}
        \frac{||x_1^{t}||^2}{\eta_1} + \frac{||x_2^{t}||^2}{\eta_2} &= \frac{||x_1^{0}||^2}{\eta_1} + \frac{||x_2^{0}||^2}{\eta_2} +\langle x_1^{0}, Ax_2^{0}\rangle - \langle x_1^{t}, Ax_2^{t}\rangle\\
        &\leq  \frac{||x_1^{0}||^2}{\eta_1} + \frac{||x_2^{0}||^2}{\eta_2} +\langle x_1^{0}, Ax_2^{0}\rangle+||x_1^t||\cdot ||Ax_2^t|| \label{eqn:Cauchy}\\
        &\leq  \frac{||x_1^{0}||^2}{\eta_1} + \frac{||x_2^{0}||^2}{\eta_2} +\langle x_1^{0}, Ax_2^{0}\rangle+||A||\cdot ||x_1||\cdot ||x_2|| \label{eqn:L2Norm}\\
        &\leq  \frac{||x_1^{0}||^2}{\eta_1} + \frac{||x_2^{0}||^2}{\eta_2} +\langle x_1^{0}, Ax_2^{0}\rangle+\frac{\sqrt{\eta_1\eta_2}||A||}{2}\left( \frac{||x_1^{t}||^2}{\eta_1}+\frac{||x_2^{t}||^2}{\eta_2}\right) \label{eqn:Youngs}
    \end{align}
    where (\ref{eqn:Cauchy}) is the Cauchy-Schwarz inequality, (\ref{eqn:L2Norm}) follows from the definition of the $\ell_2$ matrix norm, and $(\ref{eqn:Youngs})$ follows since $(\sqrt{\eta_2/\eta_1}||x_1||-\sqrt{\eta_1/\eta_2}||x_2||)^2\geq 0$. 
    Rearranging terms yields (\ref{eqn:upperbound}). 
    
    Finally, to show that $x_i^t$ is bounded, observe that $\sqrt{\eta_1 \eta_2}\leq \tfrac{2}{||A||} \Rightarrow 1 - \frac{\sqrt{\eta_1 \eta_ 2} \|A\|}{2} \geq 0$ and
    \begin{align}
         ||x_1^{t}||^2  \leq \frac{ \langle x_1^0,Ax_2^0\rangle + \tfrac{||x_1^{0}||^2}{\eta_1} + \tfrac{||x_2^{0}||^2}{\eta_2}}{\frac{1}{\eta_1} - \sqrt{\frac{\eta_2 }{\eta_1}}\frac{\|A\|}{2}}.
    \end{align}
    Symmetrically, 
    \begin{align}
         ||x_2^{t}||^2\leq \frac{ \langle x_1^0,Ax_2^0\rangle + \tfrac{||x_1^{0}||^2}{\eta_1} + \tfrac{||x_2^{0}||^2}{\eta_2}}{\frac{1}{\eta_2} - \sqrt{\frac{\eta_1 }{\eta_2}}\frac{\|A\|}{2}}
         \,,
    \end{align}
    thereby completing the proof of the theorem. 
\end{proof}
This theorem is enough to insure that~\eqref{eqn:GDAlternating} has bounded orbits in order to satisfy the hypothesis of Theorem~\ref{thm:poincar_rec}. However, it is worth noting that with the same proof technique we can derive a lower bound on a weighted sum of the norms of each agent's strategies.
\begin{theorem}\label{thm:lower_bound} If both agents use (\ref{eqn:GDAlternating}) with $||x_1^0||^2+||x_2^0||^2>0$ and fixed step-sizes such that $\sqrt{\eta_1 \eta_2 }\leq \frac{2}{\|A\|}$, then their strategies are bounded away from the equilibrium $(\mathbf{0},\mathbf{0})$. Formally, 
\begin{equation}
   \left(1 + \frac{\sqrt{\eta_1 \eta_ 2} \|A\|}{2}\right) \left(\frac{||x_1^{t}||^2}{\eta_1} + \frac{||x_2^{t}||^2}{\eta_2}\right) \geq   \left(1 - \frac{\sqrt{\eta_1 \eta_ 2} \|A\|}{2}\right) \left(\frac{||x_1^{0}||^2}{\eta_1} + \frac{||x_2^{0}||^2}{\eta_2}\right).\label{eqn:lowerbound}
\end{equation}
\end{theorem}
\begin{proof}
The proof follows identically to Theorem \ref{thm:bounded_orbits} after replacing the Cauchy-Schwarz inequality $-\langle x_1^t, Ax_2^t \rangle \leq ||x_1^t||\cdot ||Ax_2^t||$ with the Cauchy-Schwarz inequality $\langle x_1^t, Ax_2^t \rangle \leq ||x_1^t||\cdot ||Ax_2^t||$. 
By taking $\eta_1$ and $\eta_2$ sufficiently small and $||x_1^0||^2+||x_2^0||^2>0$, the right hand side of \ref{eqn:lowerbound} is positive and $(x_1^t,x_2^t)$ is bounded away from $(\mathbf{0},\mathbf{0})$. 
\end{proof}

Together, Theorems.~\ref{thm:bounded_orbits} and \ref{thm:lower_bound} indicate that (\ref{eqn:GDAlternating}) approximately preserves the energy $\nicefrac{||x_1||^2}{\eta_1}+\nicefrac{||x_2||^2}{\eta_2}$ as depicted in Figures \ref{fig:Ellipse} and \ref{fig:VerletCycles}. The smaller $\eta_1\eta_2$ is, the closer to a circle the trajectories are. 
This energy preservation does not occur for~\eqref{eqn:GDSimultaneous} as illustrated in Figure~\ref{fig:EulerDiverges}.

\subsection{Volume Preservation of Alternating Play}\label{sec:volume}

\begin{figure}[!htb]
    \centering
    \begin{subfigure}{.5\textwidth}
          \centering\captionsetup{justification=centering}
          \includegraphics[scale=.45]{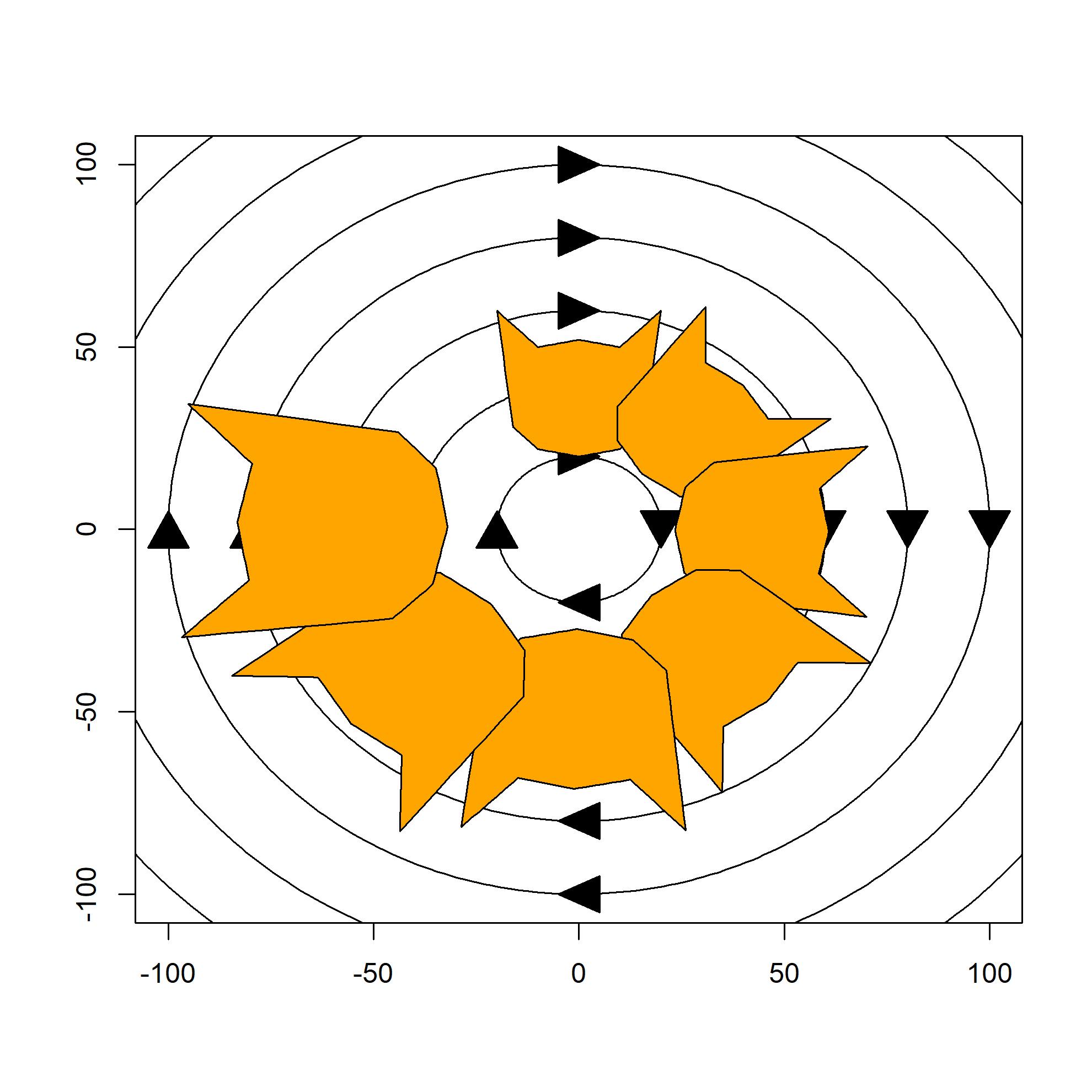}
          \caption{Volume expands when strategies are\\ updated with (\ref{eqn:GDSimultaneous}).}
          \label{fig:EulerVolume}
    \end{subfigure}%
    \begin{subfigure}{.5\textwidth}
          \centering\captionsetup{justification=centering}
          \includegraphics[scale=.45]{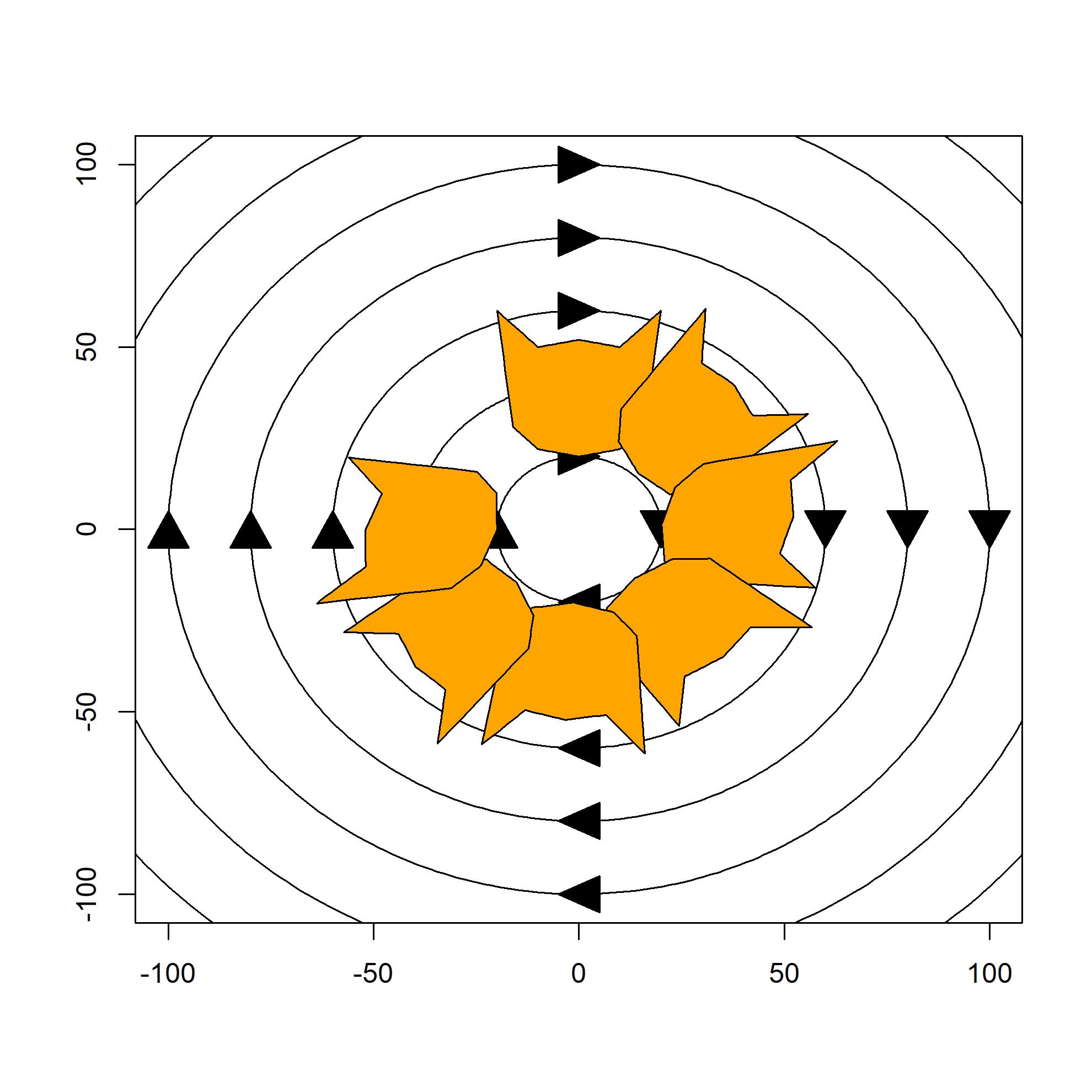}
          \caption{Volume is preserved when strategies are\\ updated with (\ref{eqn:GDAlternating}).}
          \label{fig:VerletVolume}
    \end{subfigure}
    \caption{A collection of strategies (a cat) updated by $0, 4, 8, 12, 16, 20,$ and $24$ iterations of (\ref{eqn:GDSimultaneous}) and (\ref{eqn:GDAlternating}) with $x_1,x_2\in \mathbb{R}$, $A=[1]$, and $\eta_1=\eta_2=1/5$. 
    The circles denote level sets of $||x_1||^2+||x_2||^2.$}
    \label{fig:test}
\end{figure}
In this section, we show that the transformation~\eqref{eqn:GDAlternating}  is volume preserving (Definition~\ref{def:volume_preserving}) as depicted in Figure \ref{fig:EulerVolume}. 
This is in contrast to (\ref{eqn:GDSimultaneous}) which expands (see Figure \ref{fig:EulerVolume} and \citep{cheung2019vortices}). 
To show this result, we make use of the following Theorem from \citep{rudin1987real}.

\begin{theorem}[\cite{rudin1987real} Theorem 7.26]
\label{thm:CitedVolume}
Let ${\cal X}$ be an open set in  $\mathbb{R}^k$ and $T : {\cal X} \rightarrow \mathbb{R}^k$  be an injective differentiable function with continuous partial derivatives, the Jacobian of which is non-zero for every $x\in {\cal X}$. Then for any real-valued, compactly supported, continuous function $f$, with support contained in $T({\cal X})$,
\begin{align}
\int_{T({\cal X})}f({v})d{v} =\int_{\cal X}f(T(x ))|det(J_T)({x})|d{x}.
\end{align}
In particular, taking $f(v)=1$, 
\begin{align}
\int_{T({\cal X})}d{v} =\int_{\cal X}|det(J_T)({x})|d{x}
\end{align}
and $T$ is volume preserving if $T$ is continuous differentiable, injective, and $|det(J_T)|=1$. 
\end{theorem}

Given $T:\mathbb{R}^n \to \mathbb{R}^n$, the Jacobian of the transformation is $J_T= \left[ \partial T / \partial x_1  \cdots \partial T/\partial x_n \right]$.  
The determinant of the Jacobian of~\eqref{eqn:GDAlternating} is simple to compute once noticed that~\eqref{eqn:GDAlternating} can be written as the composition of two transformations that have block triangular Jacobians. This leads to the following volume preservation theorem.

\begin{theorem}[Volume Preservation]
\label{thm:FTRL_vol}
    (\ref{eqn:GDAlternating}) is volume preserving for any step-sizes and any measurable set of initial conditions. 
\end{theorem}

\begin{proof}
    (\ref{eqn:GDAlternating}) can be written as the following two-stage update where agent 1 first updates her strategy:
    \begin{equation}\tag{Stage 1}\label{eqn:Stage1}
        \begin{aligned}
            x_1^{t+1/2}&= x_1^t + \eta_1 A x_2^t\\
            x_2^{t+1/2}&= x_2^t
        \end{aligned}
    \end{equation}
    followed by agent 2 updating his strategy: 
    \begin{equation}\tag{Stage 2}\label{eqn:Stage2}
        \begin{aligned}
            x_1^{t+1}&= x_1^{t+1/2}\\
            x_2^{t+1}&= x_2^{t+1/2}- \eta_2 A^\intercal x_1^{t+1/2}.
        \end{aligned}
    \end{equation}
    The red circles in Figure \ref{fig:Ellipse} refer to (\ref{eqn:Stage1}) while the blue triangles refer to (\ref{eqn:Stage2}). 
    To show (\ref{eqn:GDAlternating}) is volume preserving, it suffices to show that (\ref{eqn:Stage1}) and (\ref{eqn:Stage2}) are volume preserving.
    Both arguments are identical and we show the result only for (\ref{eqn:Stage1}).
    
    By Theorem \ref{thm:CitedVolume}, it suffices to show (\ref{eqn:Stage1}) is continuously differentiable, injective, and has a Jacobian with determinant equal to one.  
    Trivially, (\ref{eqn:Stage1}) is continuously differentiable.
    Next, we show it is injective. 
    
Suppose $(y_1^t,y_2^t)$ map to the same $(x_1^{t+1/2}, x_2^{t+1/2})$ when updated with (\ref{eqn:Stage1}), i.e., 
    \begin{align}
        x_1^{t+1/2}&=x_1^t+Ax_2^t=y_1^t+Ay_2^t,\\
        x_2^{t+1/2}&=x_2^t=y_2^t.
    \end{align}
     Combining both equalities yields $(x_1^t,x_2^t)=(y_1^t,y_2^t)$ and  (\ref{eqn:Stage1}) is injective.  
  
    Next, we show that the determinant of the Jacobian in (\ref{eqn:Stage1}) is one. 
    The Jacobian is 
\begin{align}
    J_1=\left[ 
        \begin{array}{c c} 
            I_{k_1}    &   \eta_1 A \\
            0                       &   I_{k_2}
        \end{array}    
    \right]
\end{align} 
and $\det(J_1)= \det(I_{k_1}) \cdot \det (I_{k_2})=1$. 
(\ref{eqn:Stage1}) satisfies all three conditions and therefore is volume-preserving thereby completing the proof of the theorem.
\end{proof}


\section{Conclusion}
\label{sec:Conclusion}

We study a natural implementation of gradient descent dynamics in unconstrained zero-sum games.
In this implementation, the max and min agent take turns updating their strategies after observing the behavior of their opponent. This dynamic has remarkable properties. First, agents have bounded regret. In fact, this is true not only in zero-sum games but in any general game and online optimization setting.  Moreover, in the max-min optimization setting the agents' strategies remain bounded for all time and the dynamics preserve volume. In combination these last two properties imply recurrence, i.e., that the orbits cycle back infinitely often arbitrarily close to their initial conditions. Such advantageous properties were formerly only known for continuous-time dynamics (e.g., \citep{Soda14,GeorgiosSODA18}) and moreover are not true for simultaneous gradient descent-ascent updates, which is divergent away from equilibrium  \citep{BaileyEC18} and in fact, formally chaotic \citep{cheung2019vortices}.  

At its core, our approach is based on recent research advances that enable connections between traditionally separate areas such as game theory, online optimization, Hamiltonian dynamics and numerical analysis. Specifically, \cite{BaileyAAMAS19} show a formal interpretation of continuous time dynamics in games as Hamiltonian systems. Based on this connection, and the numerous advantageous properties of the continuous-time dynamics it make sense to try to discretize them in a way that mimics the continuous dynamics to a high level of accuracy. From numerical analysis, it is well known that the primitive Euler approximation, i.e. the standard online optimization algorithms are poor approximations. Instead, more elaborate tools have been developed, e.g., symplectic integrators that satisfy the volume preserving property of Hamiltonian dynamics as well as other advantageous properties, e.g., approximate energy preservation (see \citep{hairer2006geometric,Hairer05}). Such symplectic integrator schemes are thus widely applied to the calculations of long-term evolution of Hamiltonian systems, e.g., planetary, molecular dynamics, a.o. Alternating gradient descent-ascent, corresponding to a symplectic integration technique known as leapfrogging (Verlet integration), is thus bringing this point of view to game dynamics. We hope that this link will allow for the development of new exciting results and a more exhaustive exploration of the connections between physics, online optimization, game theory, chaos theory, and numerical analysis.

\subsection*{Acknowledgments} %
\label{par:paragraph_name}

 James P. Bailey and Georgios Piliouras acknowledge  MOE AcRF Tier 2 Grant 2016-T2-1-170,
  grant PIE-SGP-AI-2018-01 and NRF 2018 Fellowship NRF-NRFF2018-07.
Gauthier Gidel research is partially supported by NSERC Discovery Grant RGPIN-2017-06936 and by a Borealis AI graduate fellowship.

\bibliographystyle{abbrvnat} 
\bibliography{sigproc2,IEEEabrv,Bibliography,refer}

\end{document}